\documentclass[aps, reprint,numerical,a4paper]{revtex4-1}
\usepackage{graphicx}
\usepackage{dcolumn}
\usepackage{bm}
\usepackage{color}
\usepackage{amssymb}
\usepackage{amsmath}
\usepackage{mathrsfs}
\usepackage{epstopdf}
\usepackage{natbib}

\newcommand{\vel}{v}
\newtheorem{theorem}{Theorem}
\newtheorem{lemma}{Lemma}

\newtheorem{conjecture}{Conjecture}
\newenvironment{proof}{\textit{Proof.} }{$\blacksquare$}

\begin{document}

\preprint{AIP/123-QED}

\title{Deterministic Brownian motion generated from differential delay equations}

\author{Jinzhi Lei}
\affiliation{Zhou Pei-Yuan Center for Applied Mathematics, Tsinghua University, Beijing 100084, China}

\author{Michael C. Mackey}
\affiliation{Departments of Physiology, Physics, and Mathematics, and Centre for Applied Mathematics in Bioscience and Medicine (CAMBAM), McGill University, 3655 Promenade Sir William Osler,
Montr\'{e}al, QC, Canada H3G 1Y6}
\date{\today}                                           

\begin{abstract}
This paper addresses the question of how Brownian-like motion can arise from the solution of a deterministic differential delay equation.  To study this we analytically study the bifurcation properties of an apparently simple differential delay equation and then numerically investigate the probabilistic properties of chaotic solutions of the same equation. Our results show that solutions of the deterministic equation with randomly selected initial conditions  display a Gaussian-like density for long time, but the densities are supported on an interval of finite measure.   Using these chaotic solutions as velocities, we are able to produce Brownian-like motions, which show statistical properties akin to those of a classical Brownian motion over both short and long time scales.  Several conjectures  are formulated for the probabilistic properties of the solution of the differential delay equation. Numerical studies suggest that these conjectures could be ``universal'' for similar types of ``chaotic'' dynamics,  but we have been unable to prove this.
\end{abstract}

\pacs{05.40.Ca,05.40.Jc,05.45.Ac}

\keywords{Brownian motion, central limit theorem, differential delay equation}

\maketitle

\section{Introduction}

In 1828, Robert Brown reported his observations of the apparently erratic and unpredictable movement of small particles suspended in water, a phenomena now known as ``Brownian motion''. Almost three-quarters of a century later, a theoretical (and essentially molecular) explanation of this macroscopic motion was given by Einstein, in which Brownian motion is attributed to the summated effect of a vary large number of tiny impulsive forces delivered to the macroscopic particle being observed \cite{Einstein:05} (A nice English translation of this, and other works of Einstein on Brownian motion can be found in F\"{u}rth \cite{Furth:56}). Brownian motion has played a central role in the modeling of many random behaviors in nature and in stochastic analysis, and formed the basis for the development of an enormous branch of mathematics centered around the theory of Wiener processes.

Since Brownian motion is typically explained as the summated effect of many tiny random impulsive forces, it is of interest to know if and when Brownian motion can be produced from a deterministic process (also termed as \textit{deterministic Brownian motion}) \textit{without} introducing the assumptions typically associated with the theory of random processes. Studies starting from this premise have been published in the past several decades, and there are numerous investigations that have documented the existence of Brownian-like motion from deterministic dynamics, both in discrete time maps and flows \cite{Beck91,Klages02,Trefan,Gaspard,Lasota08,Romero02,Romero05}. These models have included the motion of a particle subjected to a deterministic but chaotic force (also known as \textit{microscopic chaos}) \cite{Beck91,Trefan}, or a many-degree-of-freedom Hamiltonian \cite{Romero02,Romero05}. Experimental evidence for deterministic microscopic chaos was reported in \cite{Gaspard} by the observation of Brownian motion of a colloidal particle suspended in water (cf. \cite{Briggs01} for a more tempered interpretation, and \cite[Chapter 18]{Mazo} for other possible interpretations of experiments like these).

Several investigators have shown that a Brownian-like motion can arise when a   particle is subjected to impulsive kicks, whose dynamics are modeled by the following equations \cite{Beck91,Chew:02, Mackey06}
\begin{equation}
\label{eq:dc}
\left\{
\begin{array}{rcl}
 \dfrac{d x}{d t} &=& \vel\\
m \dfrac{d \vel}{d t} &=& - \gamma \vel + f(t).
\end{array}\right.
\end{equation}
In equation \eqref{eq:dc}, $f$ is taken to be a fluctuating ``force'' consisting of a sequence of delta-function like impulses given by, for example,
\begin{equation}
\label{eq:dc-1}
f(t) = m \kappa \sum_{n=0}^\infty \xi(t) \delta(t - n\tau),
\end{equation}
and $\xi$ is a ``highly chaotic'' deterministic variable generated by $\xi(t+\tau) = T(\xi(t))$, where $T$ is an {\it exact} map or semi-dynamical system, e.g. the tent map on $[-1, 1]$ (for more discussions and terminologies see \cite{Lasota08,Mackey06} and references therein).
In the equations \eqref{eq:dc}-\eqref{eq:dc-1}, the impulsive forces are described by $\xi(t)\delta(t-n\tau)$, which are assumed to be instantaneously effective and independent of the velocity $\vel(t)$. Dynamical systems of the form \eqref{eq:dc} have received extensive attention, and are known to be able to generate a Gaussian diffusion process \cite{Beck91,Chew:02,Mackey06,Shimizu90,Shimizu93}.

In this study, we sought an alternative continuous time description  of the ``random force'' $f(t)$, which was assumed to depend on the state (velocity) of a particle, but with a lag time $\tau$, i.e.,
\begin{equation}
f(t) = F(v(t-\tau)),
\end{equation}
and where $F$ has the appropriate properties to generate chaotic solutions.  Thus, we consider the  following differential delay equation
\begin{equation}
\label{eq:dde-0}
\begin{array}{l}
\left\{\begin{array}{rcl}
\dfrac{d x}{dt}&=& \vel\\
m\dfrac{d \vel}{dt} &=& - \gamma\vel + F(\vel(t-\tau)),
\end{array}\right.\\
\\
\quad v(t)=\phi(t), \,\,-\tau \leq t \leq 0,
\end{array}
\end{equation}
where $\phi(t)$ denotes the initial function which must always be specified for a differential delay equation.  The second equation in \eqref{eq:dde-0} is known to have chaotic solutions for some forms of the nonlinear function $F$, for example see \cite{Heiden82,Berre,Dorizzi,Ikeda87,Mackey77}. In these cases, the force $F(v(t-\tau))$ is certainly deterministic but is also  unpredictable (in practice, but not in principle)  given knowledge of the initial function. In this paper,  \textit{we will examine how a Brownian motion can be produced by the differential delay equation \eqref{eq:dde-0}}. In particular, \textit{we investigate the statistical properties of the velocity $\vel(t)$, and $\langle [\Delta x(t)]^2\rangle$, the  mean square displacement (MSD), of the solutions defined by \eqref{eq:dde-0}}. We note that unlike the equation \eqref{eq:dc}, which is linear and non-autonomous, the equation \eqref{eq:dde-0} is a nonlinear autonomous system. Numerical simulations have shown that the second equation in \eqref{eq:dde-0} can generate processes with a Gaussian-like distribution \cite{Dorizzi,Heiden82}. Nevertheless, to the best of our knowledge, there is no analytic proof for the existence of Brownian motion based on the differential delay equation \eqref{eq:dde-0}.

We first make some observations about the second equation in \eqref{eq:dde-0} which determines the dynamics of the velocity. A simple form of the ``random'' force is binary and fluctuates between $\pm f_0$, for instance, given by
\begin{equation}
F(v) = 2 f_0\left[H(\sin (2\pi \beta v)) - \dfrac{1}{2}\right],
\end{equation}
where $H$ is the Heavyside step function, i.e.,
\begin{equation}
H(\vel) = \left\{\begin{array}{ll}
0\qquad & \mathrm{for}\ \vel<0\\
1 & \mathrm{for}\ \vel\geq 0.
\end{array}\right.
\end{equation}
Then we have following equation
\begin{equation}
\label{eq:ex-2}
\begin{array}{l}
\dfrac{d\vel}{dt} = -\gamma\vel + 2  \left\{H(\sin(2 \pi \beta \vel(t-1))) - \frac{1}{2}\right\}.\\
v(t)=\phi(t), \,\,-\tau \leq t \leq 0.
\end{array}
\end{equation}
Here and later we always assume the mass $m=1$ and $f_0 = 1$ which can be achieved through the appropriate scaling.
The delay differential equation \eqref{eq:ex-2} with a binary ``random force'' can be solved iteratively by the method of steps \footnote{A solution of \eqref{eq:ex-2} is associated with a time sequence $t_0 < t_1 < \cdots< t_n< \cdots$,
which is defined such that $\sin(2\pi\beta \vel(t)) \geq 0$ when $t\in [t_{2k}, t_{2k+1})$, and
$\sin(2\pi\beta \vel(t)) < 0 $ when $t\in [t_{2k-1}, t_{2k})$. Furthermore, if the sequence $(t_0, \cdots, t_n)$ is known, then the solution $\vel(t)$ when $t\in (t_n, t_n + 1)$ can be obtained explicitly, and therefore, $t_{n+1}$, which is defined as $\sin(2\beta \vel(t_{n+1})) = 0$,  is determined by $(t_0, \cdots, t_n)$. Once we obtain the entire sequence $\{t_n\}$, the solution of \eqref{eq:ex-2} consists of exponentially increasing or decreasing segments on each interval $[t_n, t_{n+1}]$. Nevertheless, the nature and properties of the  map $t_{n+1} = F_n(t_0, t_1,\cdots, t_{n})$ is still not characterized and has defied analysis to date.
}. Despite its simplicity, it can display behaviors similar to a random process. An example solution of \eqref{eq:ex-2} is shown in Figure \ref{fig:ex}, which looks like noisy.

\begin{figure}[htbp]
\centering
\includegraphics[width=8cm]{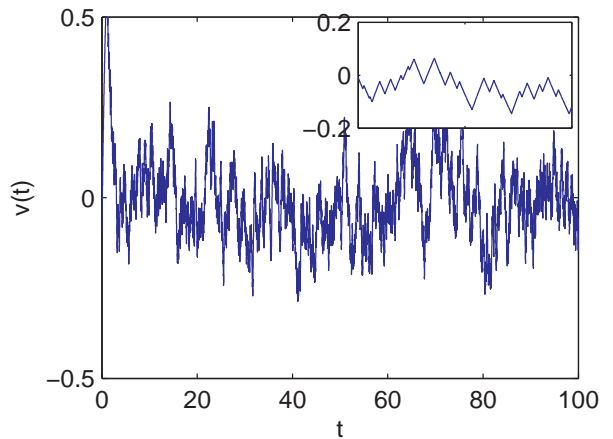}
\caption{(Color online) A sample solution of \eqref{eq:ex-2} with $\beta=10, \gamma = 1$, and an initial function $\phi(t) \equiv -0.1, t \in [-1,0]$. The rectangular inset shows the solution segment for $98 \leq t \leq 100$.}
\label{fig:ex}
\end{figure}
The ``random force'' in \eqref{eq:ex-2} is discontinuous and gives a continuous zigzag velocity curve (c.f. inset in Figure \ref{fig:ex}). In this paper, we will instead study an analogous different differential delay equation
\begin{equation}
\label{eq:1}
\begin{array}{l}
\dfrac{d \vel}{dt} = - \gamma\vel + \sin(2\pi\beta \vel(t-1)), \\
v(t)=\phi(t), \,\,-1 \leq t \leq 0.
\end{array}
\end{equation}
In \eqref{eq:1}, the parameter $\beta$ measures the frequency of the nonlinear function, and will turn out to be an essential parameter in the present study.  Note that one can re-scale and translate the variables such that equation \eqref{eq:1} can be rewritten as
\begin{equation}
\label{eq:2}
\begin{array}{l}
\dfrac{d \vel}{dt}=  - \vel + \mu \sin (\vel(t-\tau) - x_0), \\
v(t)=\phi(t), \,\,-\tau \leq t \leq 0.
\end{array}
\end{equation}
Equation \eqref{eq:2} (also known as Ikeda equation) was proposed by Ikeda   \textit{et. al.} to model a passive optical bistable resonator system, and shows chaotic behaviors  at
particular parameters such as $\mu = 20, x_0 = \pi/4$ and $\tau = 5$ \cite{Ikeda80,Ikeda87}.

In this paper,  we will study the dynamical properties of the solutions of \eqref{eq:1}, both analytically and numerically. We focus in particular on the probabilistic properties of the chaotic solutions.   Then we investigate chaotic solutions of
\begin{equation}
\label{eq:bm}
\begin{array}{l}
\left\{\begin{array}{rcl}
\dfrac{d x}{dt}&=& \vel\\
\dfrac{d \vel}{dt} &=& - \gamma\vel + \sin(2\pi\beta \vel(t-1)),
\end{array}\right.\\
\\
\quad v(t)=\phi(t), \,\,-1 \leq t \leq 0,
\end{array}
\end{equation}
and characterizing the statistical properties as completely as we can. The main result is to show that the equation \eqref{eq:bm} can reproduce experimentally observed data of Brownian motion over a wide range of time scales, in spite of the fact that the evolution equation is deterministic. Therefore, deterministic Brownian motion can be generated from the equation \eqref{eq:bm}.

The outline of the rest of this paper is as follows.  We first perform a bifurcation analysis for equation \eqref{eq:1} in Section \ref{sec:2}. In Section \ref{sec:3} we study the probabilistic properties of the
chaotic solutions numerically. In Section \ref{sec:4}, we numerically examine the dynamics of the chaotic solutions of \eqref{eq:bm} and compare our results with recent experimental measurements of the motion of a Brownian particle \cite{Raizen}.  Section \ref{conjectures} presents five conjectures based on our studies that we have been unable to prove, but which we believe to be true. These conjectures indicate a possible direction for the analytical proof of the existence of deterministic Brownian motion from differential delay equation \eqref{eq:dde-0}. Finally, we conclude the paper with discussion and conclusions in Section \ref{sec:5}.

\section{Bifurcation analysis}
\label{sec:2}

In this section, we commence our study by performing a bifurcation analysis for equation \eqref{eq:1}. We always assume $\gamma = 1$.

The bifurcation structure of the Ikeda equation has been studied several times from different perspectives \cite{Hopf82,Mandel83,Nardone86,Chow92,Hale94,Ernex04}.  Here, we present a complete picture (see Theorem \ref{th:bif} below) for the bifurcation structure of the equation \eqref{eq:1}, which has not, to the best of our knowledge,  appeared previously.

\subsection{Steady state solutions}
\label{sec:2.1}

The steady states of equation \eqref{eq:1} are given by the solutions of
\begin{equation}
\label{eq:ss}
\vel = \sin (2\pi \beta \vel).
\end{equation}
When $\beta \leq 1/(2\pi)$,  equation \eqref{eq:ss}  has only one real solution, namely $\vel = 0$. When $\beta > 1/(2\pi)$,  \eqref{eq:ss} has $(4 [\beta]+1)$ real solutions where $[\beta]$ denotes the integer part of $\beta$. These solutions are separated by critical points which are given by the roots of
 $$
 1  = 2\pi\beta \cos(2\pi \beta \vel),
 $$
 {\it i.e.},
 $$
 \dfrac{k}{\beta}\pm \dfrac{1}{2\pi\beta}\arccos\left (\dfrac{1}{2\pi \beta}\right ),\quad -[\beta]\leq k\leq  [\beta].$$

 Let $\vel^*$ be a steady state of \eqref{eq:1}.  Linearization of \eqref{eq:1} around $\vel = \vel^*$ gives
\begin{equation}
\label{eq:lin}
\dfrac{d \tilde{\vel}}{dt} = - \tilde{\vel} + 2\pi\beta \cos(2\pi \beta v^*) \tilde{\vel}(t-1).
\end{equation}
Thus, the steady state solution $\vel(t) \equiv \vel^*$ is locally stable if and only if
\begin{equation}
\begin{array}{c}
\sec \omega \leq 2\pi \beta \cos (2\pi \beta \vel^*) \leq 1,\\
\mathrm{where}\ \omega + \tan \omega = 0,\ \omega\in (0,\pi).
\end{array}
\end{equation}
In particular, if $2\pi \beta \cos (2\pi \beta \vel^*) = \sec\omega$, the linearized equation \eqref{eq:lin} has a pair of complex conjugate
eigenvalues, and therefore has a periodic solution  with frequency $\omega$. In this case, $(\beta, \vel^*)$ is a
\textit{Hopf bifurcation  point}
of \eqref{eq:1}.  (Throughout this paper we will use the notation $(\beta, \vel^*)$ where $\beta$ is the bifurcation parameter and $\vel^*$ is the bifurcation point.)  Figure \ref{fig:ss} graphically displays the steady states for $0\leq \beta \leq 6$.

\begin{figure}[htbp]
\centering
\includegraphics[width=8cm]{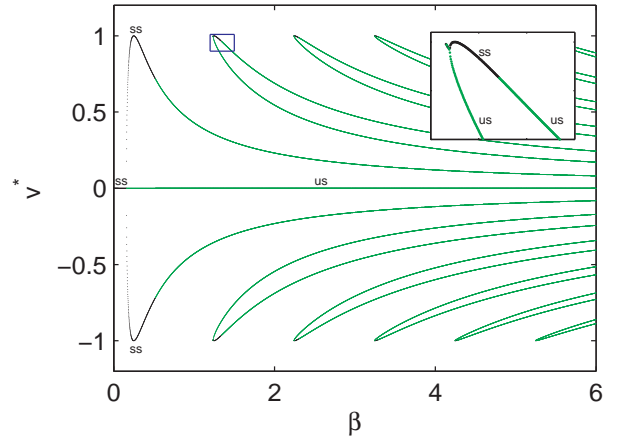}
\caption{(Color online) Steady states $v^*$ of \eqref{eq:1} for different values of $\beta$ (here $\gamma=1$). The inset shows the detail in the small rectangular area. Black denotes locally stable
steady states (ss), while green (light gray) denotes unstable steady states (us).}
\label{fig:ss}
\end{figure}

We will first state the following  Lemma before giving the main results of our bifurcation analysis.

\begin{lemma}
\label{le:0}
For any $\alpha \in (-\infty, 0)\cup \left (\frac{1}{2\pi}, +\infty \right )$, let
\begin{equation}
p(\beta) = 2\pi \beta \sqrt{1 - \left (\dfrac{\alpha}{2\pi \beta}\right )^2} - \arccos \left ( \frac{\alpha}{2\pi \beta} \right ).
\end{equation}
Then for any $k\in \mathbb{N}^*$, the equation
\begin{equation}
\label{eq:pb}
2\pi k = p(\beta)
\end{equation}
has a unique solution $\beta \geq {|\alpha|}/{2\pi}$. In particular, when $\alpha  = 1$ and $k = 0$, we have $\beta =  {1}/{2\pi}$.
\end{lemma}
\begin{proof}
Since $p  (\frac{|\alpha|}{2\pi}  ) \leq 0$, and when $\beta > \frac{|\alpha|}{2\pi}$,
$$p'(\beta) = \dfrac{4\pi \beta^2 - \alpha}{\beta \sqrt{4\pi^2 \beta^2 - \alpha}} > \dfrac{\alpha (2\pi \alpha - 1)}{\beta \sqrt{4\pi^2 \beta^2 -
\alpha}} > 0,$$
the first part of the Lemma follows. It is easy to verify that when $\alpha =1$, then $p (\frac{1}{2\pi}  ) = 0$.
\end{proof}

In the following, we define sequences $\{a_k\}$ and $\{b_k\}$ such that
\begin{equation}
\label{eq:ak0}
2\pi k = 2\pi a_k \sqrt{1 - \left (\dfrac{1}{2\pi a_k}\right )^2} - \arccos \left ( \dfrac{1}{2\pi a_k}\right ),\quad k\in \mathbb{N}^*,
\end{equation}
and
\begin{equation}
\label{eq:bk0}
2\pi k = 2\pi b_k \sqrt{1 - \left (\dfrac{\sec \omega}{2\pi b_k}\right )^2} -  \arccos \left (\frac{\sec\omega}{2\pi b_k} \right ),\quad k\in \mathbb{N}^*.
\end{equation}
Note that $a_0 = 1/2\pi$.  Furthermore, since $\sec^2\omega > 1$ and $\arccos ( \frac{\sec \omega}{2\pi\beta} ) > \arccos  ( \frac{1}{2\pi \beta} )$, we have
$$
\begin{array}{l}
2\pi \beta \sqrt{1 - \left (\dfrac{\sec \omega}{2\pi\beta}\right )^2} -  \arccos \left ( \dfrac{\sec\omega}{2\pi \beta}\right ) \\
\hspace{1.0cm}< 2\pi \beta \sqrt{1- \left (\dfrac{1}{2\pi
\beta}\right )^2} - \arccos \left ( \dfrac{1}{2\pi \beta}\right ),
\end{array}
$$
and therefore $b_k > a_k$. When $k\to\infty$, the solution of \eqref{eq:pb} is approximately
$$
\beta \simeq  k + \dfrac{1}{4} + \dfrac{\alpha (\alpha - 2)}{4\pi^2 k} + o(k^{-1}).
$$
Thus, for large $k$ we have
\begin{equation}
\label{eq:ab}
b_k - a_k  = \dfrac{(\sec\omega - 1)^2}{4\pi^2 k}+o(k^{-1}).
\end{equation}
Note that the sequences $\{a_k\}$ and $\{b_k\}$ can be ordered as
\begin{equation}
0 < a_0 < b_0 < a_1 < b_1 \cdots < a_k < b_k < \cdots.
\end{equation}
In the following Theorem, we prove that $\{a_k\}$ are saddle node bifurcation points and $\{b_k\}$ are Hopf bifurcation points of Eq. \eqref{eq:1}.

\begin{theorem}
\label{th:bif}
Consider Eq. \eqref{eq:1} and its steady state solutions. Let $\{a_k\}$ and $\{b_k\}$ be defined as above. Then:
\begin{enumerate}
\item[\textnormal{(1)}] When $0\leq\beta<a_0$, \eqref{eq:1} has only one steady state $\vel^* = 0$, and it is locally stable.
\item[\textnormal{(2)}] When $\beta = a_0$, $(\beta, \vel) = (a_0, 0)$ is a  pitchfork bifurcation point of \eqref{eq:1}.
\item[\textnormal{(3)}] When $\beta > a_0$, the steady state $\vel^* = 0$ is unstable.
\item[\textnormal{(4)}] For any $k\in \mathbb{N}^*$, let
\begin{equation}
\label{eq:yk}
y_k = \dfrac{k}{a_k} + \dfrac{1}{2\pi a_k} \arccos \left ( \dfrac{1}{2\pi a_k}\right ).
\end{equation}
Then $(\beta, \vel) = (a_k, \pm y_k)$ are saddle node bifurcation points of \eqref{eq:1}.
\item[\textnormal{(5)}] For any $k\in \mathbb{N}^*$, let
\begin{equation}
\label{eq:zk}
z_k = \dfrac{k}{b_k} + \dfrac{1}{2\pi b_k} \arccos \left (\frac{\sec \omega}{2\pi b_k}\right ).
\end{equation}
Then $(\beta, \vel) = (b_k, \pm z_k)$ are Hopf bifurcation points of \eqref{eq:1}.
\item[\textnormal{(6)}] For every $k\in \mathbb{N}_0$, and $a_k$ as defined above, there is a  function $f_k(\beta)$ which is continuous on $[a_k, \infty)$, such
that $f_k(a_k) = y_k$, and when $\beta \geq a_k$,  $\vel = \pm f_k(\beta)$ satisfies \eqref{eq:ss}, and
\begin{equation}
\label{eq:u1}
1 - 2\pi \beta \cos (2\pi \beta f_k(\beta)) \leq 0.
\end{equation}
The steady state solutions $\vel(t) \equiv  \pm f_k(\beta)$ are unstable.
\item[\textnormal{(7)}] For every $k\in \mathbb{N}_0$, and $a_k, b_k$ as defined above, there is a  function $g_k(\beta)$, which is continuous on $[a_k, \infty)$,
such that $g_k(a_k) = y_k$, $g_k(b_k) = z_k$, and when $\beta \geq a_k$, $\vel=\pm g_k(\beta)$ satisfy \eqref{eq:ss}, and
\begin{equation}
\label{eq:u2}
1 - 2\pi \beta \cos (2\pi \beta g_k(\beta)) \geq 0.
\end{equation}
Further,
\begin{enumerate}
\item[\textnormal{(a)}] When $a_k < \beta < b_k$, the steady state solutions $\vel(t) \equiv \pm g_k(\beta)$ are locally stable.
\item[\textnormal{(b)}] When $\beta > b_k$, the steady state solutions $\vel(t) \equiv \pm g_k(\beta)$ are unstable.
\end{enumerate}
\item[\textnormal{(8)}] When $\beta$ increases past the Hopf bifurcation point $b_k$, the two   steady state solutions $\vel(t)\equiv \pm g_k(\beta)$ lose
stability and generate a periodic solution, with angular frequency $\omega$.  Therefore there exists a sequence $\{c_k\}$, such that the periodic
solutions generated from the Hopf bifurcations are stable when $b_k < \beta < c_k$.
\end{enumerate}
\end{theorem}

Proof of Theorem \ref{th:bif} is given in the Appendix.

From Theorem \ref{th:bif}, when $\beta$ increases from  $0$, in any
interval $\beta \in (a_k, b_k)$, equation \eqref{eq:1} has two stable steady state solutions. The length of the intervals $(a_k, b_k)$ tends to zero as $1/k$ as $k\to\infty$.  In other situations, however, all steady states are unstable, and therefore complicated dynamical behaviors may be expected.  An exploration of the nature of these constitutes the remainder of this paper.

\subsection{Periodic solutions and chaotic attractors}

In this section, we numerically investigate the long term behavior of the solutions of equation \eqref{eq:1}. For each value of the  parameter $\beta > 0$, we solve
equation \eqref{eq:1} to obtain $100$ independent sample solutions, each with randomly selected constant initial function
$$
\vel(t) = \vel_0\quad (-1 \leq t \leq 0),
$$
where $\vel_0\in (-1,1)$ and is uniformly distributed. Each solution is obtained using Euler's method (with a time step $\Delta t = 0.001$) up to $t=500$, such that the solution reaches a stable state (either oscillatory or steady state), and the resulting data from $300\leq t \leq 500$ are used for further analysis as detailed below.

To distinguish oscillatory solutions from constant solutions in the simulation, we investigated the upper and lower bounds of $\vel(t)$ in $300 \leq t\leq 500$, denoted by $\vel_{\max}$ and $\vel_{\min}$, respectively. Therefore, a solution is considered to have approached a stable steady state if $\vel_{\max}\simeq \vel_{\min}$, and approached a stable oscillatory solution between $ \vel_{\min}$ and $\vel_{\max}$ if $\vel_{\min} \ll \vel_{\max}$.

Figure \ref{fig:bif} shows the simulation results. For each value of $\beta\in (0,3)$, there are $100$ pairs of dots from 100 initial functions, corresponding respectively to
$v_{\min}$ (red (light gray) dots) and $v_{\max}$ (blue dots) of a solution.  Stable steady state solutions are shown by the superposition of blue and red (light gray) dots.
We are interested in the oscillatory solutions, which are indicated by well separated blue and red (light gray) dots in Figure \ref{fig:bif}.  There are two
types of oscillatory solutions. Regular periodic solutions appear when $\beta$ is close to (but greater than) the Hopf bifurcation points $b_k$. The
amplitude of these regular solutions only depend on the parameter $\beta$. Irregular oscillatory solutions occur for almost all  $\beta$ ($\beta >
0.85$ except small gaps at $\beta \in (1.019, 1.033)$ and $\beta \in (1.270, 1.360)$, respectively. The amplitude of these irregular oscillatory
solutions depend on the parameter $\beta$ as well as the initial functions.

The bifurcation diagram in Figure \ref{fig:bif} shows clear evidence for multi-stability of the solutions of equation \eqref{eq:1}:
\begin{enumerate}
\item When $a_0 < \beta < b_0$, there are two stable steady states.
\item When $b_0 < \beta < c_0$, there are two stable periodic solutions.
\item When $a_k < \beta < b_k  (k\in \mathbb{N}^*)$, there are two stable state states, and stable irregular oscillations.
\item When $b_k < \beta < c_k (k\in \mathbb{N}^*)$ there are two stable periodic solutions, and stable irregular oscillations.
\end{enumerate}
Some sample solution examples for different values of $\beta$ are shown in Figure \ref{fig:sam}, and these also illustrate the existence of multi-stability of solutions dependent on the initial function.

\begin{figure}[htbp]
\centering
\includegraphics[width=10cm]{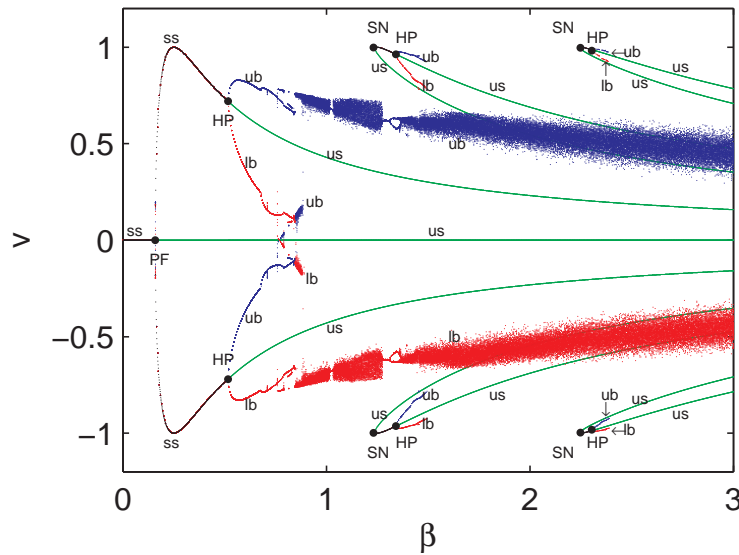}
\caption{(Color online) Bifurcation diagram of \eqref{eq:1} for different $\beta$ (here $\gamma=1$). Here PF denotes the occurrence of a pitchfork bifurcation, HP stands for a Hopf bifurcation, and SN a saddle node bifurcation. In the diagram, black dots denote   stable steady states (ss),  green denotes unstable steady states (us) (also refer to Figure \ref{fig:ss}),  blue is used to indicate the upper bound (ub) of oscillatory solutions, and red (light gray) indicates the lower bound (lb) of oscillatory solutions. }
\label{fig:bif}
\end{figure}

\begin{figure*}[htbp]
\centering
\includegraphics[width=14cm]{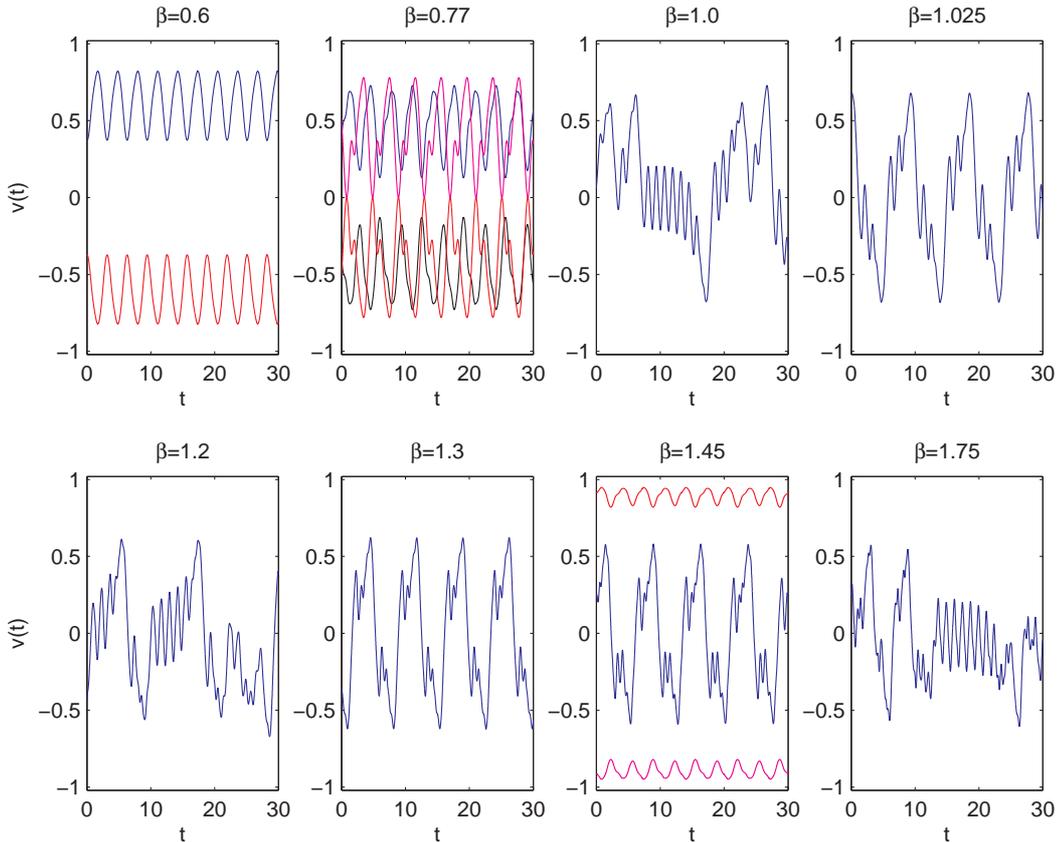}
\caption{(Color online) Sample solution segment examples for different values of $\beta$ (from small to large) as indicated in each panel. The time point $0$ in figures correspond to $t=400$ in the simulations. Each panel contains between one and four solution curves (marked by different colors and also locations), each of which corresponds to a stable oscillatory solution. For a given values of $\beta$, the initial functions of these solutions are different from each other, and taken as constant in the interval $-1 \leq  t \leq 0$.  The constant initial functions are: $\beta=0.6$: $v_0 = 0.55,   -0.55$; $\beta=0.77$: $ v_0 = 0.41129,   -0.41129,  0.630575,  -0.630575$; $\beta=1.0$: $ v_0=0.55$; $\beta=1.025$: $ v_0=0.55$; $\beta=1.2$: $ v_0=0.55$; $\beta=1.3$: $ v_0=0.55$; $\beta=1.45$: $ v_0 = 0.975,  0.1,  -0.975$; $\beta=1.75$: $ v_0=0.1$. Here $\gamma=1$.
}
\label{fig:sam}
\end{figure*}

\section{Probabilistic Properties}
\label{sec:3}

In the numerical bifurcation analysis of Section \ref{sec:2}, we have shown that when $\gamma=1$ and $\beta > 0.85$, the differential delay equation \eqref{eq:1} has
irregular oscillatory solutions that display chaotic behavior. In this section, we numerically study the probabilistic properties of these irregular solutions.  We will show in Section \ref{ss:quasigaussian} that when $\beta$ is sufficiently large, on a sufficiently long time scale these chaotic solutions behave like a noise source with a  ``truncated'' Gaussian density.

\subsection{Numerical scheme}
\label{sec:ns}

Throughout this section, the probabilistic properties of solutions of equation \eqref{eq:1} are studied numerically. In the numerical simulations, for a given set of parameters, we solve the equation \eqref{eq:1} with a randomly selected constant initial function
\begin{equation}
\label{eq:ini}
\vel(t) = \vel_0 \in (-1, 1),\quad (-1 \leq t\leq 0),
\end{equation}
where $\vel_0$ is drawn from a uniformly distributed density.
The solution $\vel(t)$ is solved using Euler's method (with a time step $\Delta t = 0.001$) up to $t=10^5$, and is sampled every $10^3$ steps to generate a time series $\{\vel_n\}$, where $\vel_n = \vel(n\times 10^3 \Delta t)$. The resulting time series of values $\{\vel_n\}$ is used to characterize the statistical properties of the solution. In particular,  we focus on the mean value
$\mu$, the upper bound $K$, the standard deviation $\sigma$, and the excess kurtosis $\gamma_2$ of the time series, which are respectively defined by
\begin{equation}
\label{eq:pp}
\begin{array}{c}
\displaystyle
\mu = \dfrac{1}{N}\sum_{n=1}^N \vel_n,\ K = \max_n |\vel_n|,\ \sigma^2 = \dfrac{1}{N}\sum_{n=1}^N (\vel_n - \mu)^2,\\
\displaystyle
\gamma_2 = \dfrac{\mu_4}{\sigma^4} - 3,\ \mathrm{where}\ \mu_4 = \dfrac{1}{N}\sum_{n=1}^N (\vel_n - \mu)^4.
\end{array}
\end{equation}
The excess kurtosis $\gamma_2$ measures the sharpness of the density of the sequence, and a value of $\gamma_2=0$ is characteristic of a normal Gaussian distribution.

In the following discussion, we will show that when $\beta$ is outside the region of
bistability, {\it i.e.},
\begin{equation}
\label{eq:defI}
 \beta \in I = \left (\dfrac{1}{2\pi},+\infty \right )\backslash \bigcup_{k=0}^\infty [a_k , c_k],
\end{equation}
the statistical properties are independent of the initial function $\vel_0$ and sampling frequencies. Therefore, the quantities defined by \eqref{eq:pp} only depend on the  parameters $\beta$ and $\gamma$, and this dependence is discussed below.

\subsection{Stationary density of solutions}\label{ss:stationary}

From the bifurcation analysis in Section \ref{sec:2}, equation \eqref{eq:1} displays bistability when $\beta \in [a_k, c_k]$ for some $k\in \mathbb{N}$. This suggests that when
\begin{equation}
\beta \in I_0 = \bigcup_{k\in \mathbb{N}} [a_k, c_k],
\end{equation}
a solution of \eqref{eq:1} with a randomly selected initial function will converge to one of the stable branches (either a steady state, or a periodic solution, or an irregular oscillatory solution) as shown by Figure \ref{fig:bif}. Thus, the stationary density of all solutions of \eqref{eq:1} is expected to be multi-modal, and therefore the limiting statistical properties of a solution (as mirrored in the density constructed along the solution trajectory) depend on the initial condition. The upper two panels of Figure \ref{fig:density}a show the multi-modal distributions of the stationary densities when $\beta = 1.25$ and $\beta = 1.35$, respectively. The results are obtained from $10^5$ independent solutions at $t = 100$, each with a randomly selected constant initial function as given by \eqref{eq:ini}.

In the following discussion, we focus on the alternative situation in which
\begin{equation}
\beta \in I = \left (\dfrac{1}{2\pi}, \infty \right )\backslash I_0.
\end{equation}
In this case, the numerical bifurcation analysis shown in Figure \ref{fig:bif} indicates that any solution of \eqref{eq:1} converges to an irregular oscillatory solution irrespective of the initial function. (Note that solutions with different initial functions will not, in general, converge to the same solution as can be seen  by multiple values of the upper and lower bounds of the oscillatory solutions in Figure \ref{fig:bif}.) Despite the fact that the solutions in such situations are not the same,  we will see that these solution trajectories share the same statistical properties.  Figure \ref{fig:density}a (black curve in the lowest panel) shows the stationary density obtained from $10^5$ independent solutions (here $\beta = 2.0 \in I$), each with a randomly selected constant initial function \eqref{eq:ini}. The result is a uni-modal density of the distribution of solution values along the trajectory. Alternatively, the same density can also be obtained through a time series $\{\vel_n\}$ of a solution with randomly selected constant initial function (Figure \ref{fig:density}a, red (light gray) curve in the lowest panel, refer to Section \ref{sec:ns} for details). Furthermore, we have found that the same density of the distribution is obtained when we choose different {\it forms} for the initial function, such as a sinusoidal initial function, or a polynomial initial function (data not shown). These results strongly suggest that these irregular solutions are ergodic in some sense, i.e., the statistical properties of one solution are the same as those  of an
ensemble of independent solutions {\footnote {It is important to realize, however, that the notion of ergodicity for an infinite dimensional semi-dynamical system, like the differential delay equations we are studying, is not well defined and has resisted all attempts to do so.}}.

In what follows, we focus on quantifying the statistical properties of these irregular
solutions.  An example of one of these irregular solutions is shown in Figure \ref{fig:density}b, with the corresponding power spectrum $w(f)$ shown in Figure \ref{fig:density}c. Note that the power spectrum is essentially flat with no predominant characteristic frequency, indicating that the solution is, indeed, chaotic.

\begin{figure}[htbp]
\centering
\includegraphics[width=8cm]{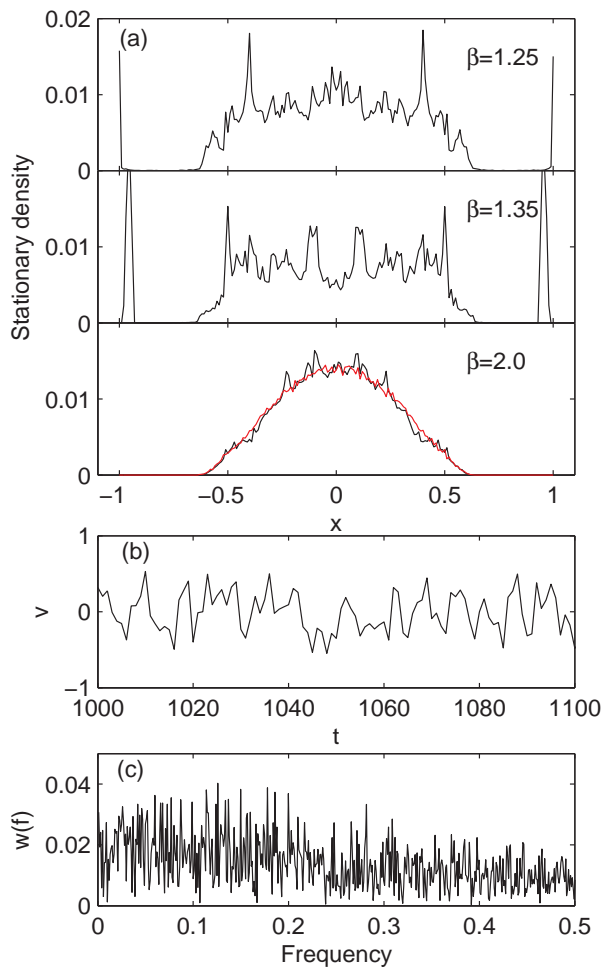}
\caption{(Color online) (a) Stationary density functions of solutions of \eqref{eq:1} for different values of $\beta$ (as shown in the panels).  The density functions are obtained from $10^5$ independent solutions, each with a randomly selected constant initial function as in Eq. \eqref{eq:ini}. The red (light gray) curve in the lowest panel is the density function obtained from a random solution trajectory with a constant initial function.   (b) Example of a segment of a single solution with $\beta = 2.0$,   and initial function $v_0 = 0.04$. (c). Power spectrum of the solution in (b). In the simulations, $\gamma = 1.0$. }
\label{fig:density}
\end{figure}

\subsection{Dependence of the statistical properties on $\beta$}
\label{sec:3.2}
Now, we take $\gamma = 1$, and $\beta \in [1, 50]$ to numerically study the statistical properties of the irregular solutions. For each $\beta$ the sampled time series $\{\vel_n\}$ of a solution $\vel(t)$ is used to obtain the mean value $\mu$, the upper bound $K$, the standard deviation $\sigma$, and the excess kurtosis $\gamma_2$ (refer Eq. \eqref{eq:pp} for these definitions).
Figure \ref{fig:parb} shows these four statistical indicators as functions of the parameter $\beta$.

Figure \ref{fig:parb}a shows the mean value as a function of $\beta$, indicating that $\mu(\beta) \simeq 0$. Figure \ref{fig:parb}b shows the bound $K$ as a function of $\beta$. The numerical results show that $K$ decreases with $\beta$, and can be accurately approximated by
\begin{equation}
K = \dfrac{1}{0.68\sqrt{\beta}+0.60}.
\end{equation}
Note that as $\beta\to\infty$, solutions of \eqref{eq:1} are bounded  and $K$ varies as $\beta^{-1/2}$. Figure \ref{fig:parb}c shows the standard deviation $\sigma$ as a function of $\beta$. The standard deviation decreases with $\beta$, and can be
fitted with
\begin{equation}
\sigma = 0.32 \beta^{-1/2}.
\end{equation}
As in the case of the upper bound, when $\beta\to\infty$, we also find that $\sigma$ varies as $\beta^{-1/2}$.

Figure \ref{fig:parb}d shows the excess kurtosis $\gamma_2$ as a function of $\beta$, and the numerical results reveal that the excess kurtosis increases with $\beta$
towards $0$, approximately as $-1/\beta$.  The negative value indicates that the distribution is platykurtic (the tail of the distribution is thinner relative to a Gaussian). This is because the solution is bounded, and therefore the tail is truncated.  Note that a larger $\beta$ means a smaller \textbf{absolute} value of the excess kurtosis, and thus that the distribution is more like a Gaussian distribution.

\begin{figure}[htbp]
\centering
\includegraphics[width=8cm]{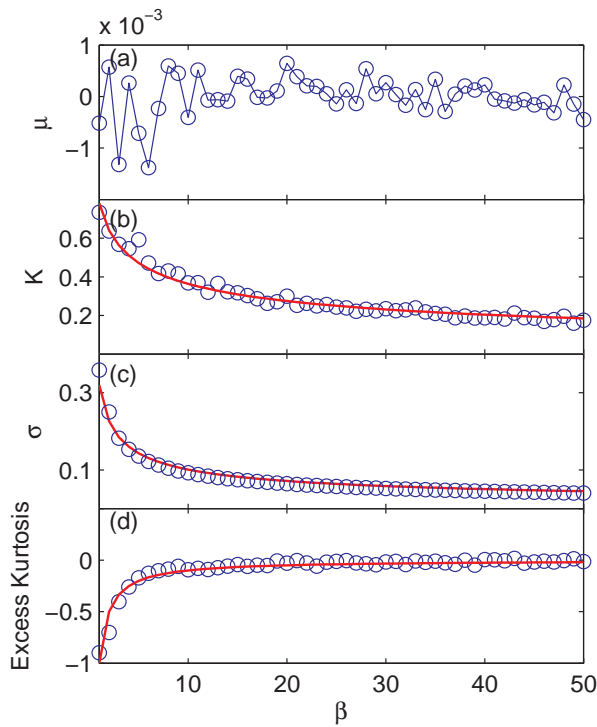}
\caption{(Color online) (a) Mean value $\mu$ as  a function of $\beta$. (b) The upper bound $K$ as a function of $\beta$. The solid curve shows the fit to $K = 1/(0.68\sqrt{\beta} + 0.60)$. (c) Standard deviation $\sigma$ as a function $\beta$.  The solid curve is the graph of  $\sigma = 0.32 \beta^{-1/2}$. (d) Excess kurtosis $\gamma_2$ as a function of $\beta$. Solid curve shows the fit to $\gamma_2 \simeq -1/\beta$. Remember that for all of these results, $\gamma = 1$,  and the initial functions are constants as in \eqref{eq:ini}.}
\label{fig:parb}
\end{figure}

\subsection{Dependence of the statistical properties on $\gamma$}
\label{sec:3.3}

We now fix $\beta = 20$ and study the dependence of the statistical properties on $\gamma$. Figure \ref{fig:pargamma} shows the
simulation results with the same statistical indicators plotted as in Figure \ref{fig:parb}.
The red (light gray) curves in Figure \ref{fig:pargamma} are fit by
\begin{eqnarray}
\label{eq:K0}
K(\beta,\gamma) &=& \dfrac{1}{\sqrt{\gamma} (0.68\sqrt{\beta} + 0.60 \sqrt{\gamma})},\\
\label{eq:sigma}
\sigma(\beta, \gamma) &=& \dfrac{0.32}{\sqrt{\beta \gamma}},\\
\label{eq:kur}
\gamma_2(\beta,\gamma) &=& - \dfrac{\gamma}{\beta}.
\end{eqnarray}

The functions \eqref{eq:K0}-\eqref{eq:kur} give the general dependence of the statistical indicators with equation parameters $\beta$ and $\gamma$, and are obtained as follows. First, we rescale   equation \eqref{eq:1} by introducing $u = \gamma \vel, \beta' = \beta/\gamma$. Then $u(t)$ satisfies
\begin{equation}
\label{eq:u}
\dfrac{d u}{d t} =\gamma ( - u + \sin (2\pi \beta' u(t - 1))).
\end{equation}
The statistical indicators $K$, $\sigma$ and $\gamma_2$ for solutions $u(t)$ of \eqref{eq:u} are independent of $\gamma$ (data not shown), and depend on $\beta'$ through the same functions as in Figure \ref{fig:parb}.  Therefore, we obtain the functions \eqref{eq:K0}-\eqref{eq:kur} with the scaling $\vel = u/\gamma$.

\begin{figure}[htbp]
\centering
\includegraphics[width=8cm]{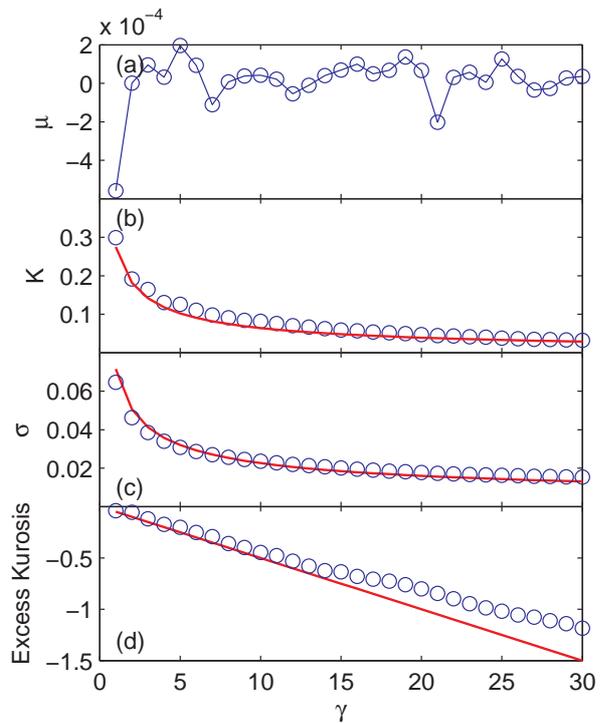}
\caption{(Color online) (a) Mean value $\mu$ as  a function of $\gamma$. (b) The upper bound $K$ as a function of $\gamma$.  (c) Standard deviation $\sigma$ as a
function $\gamma$.  (d) Excess kurtosis as a function of $\gamma$.  Solid curves in (b)-(d) show the fits \eqref{eq:K0}-\eqref{eq:kur}. Here $\beta = 20$  and the initial functions are constants as in Eq. \eqref{eq:ini}. Results for $\gamma > 30$ are not shown due to numerical instability.}
\label{fig:pargamma}
\end{figure}

\subsection{Correlation function}\label{ss:correlation}

Next, we investigate the correlation function of a solution of the differential delay  equation \eqref{eq:1}.  The normalized correlation function of a solution is
defined as
\begin{equation}
C(r) = \lim_{T\to\infty}\dfrac{\int_0^T \vel(t)\vel(t+r) d t}{\int_0^T \vel(t)^2 d t}.
\end{equation}

Figure \ref{fig:corr}a shows the correlation function $C(r)$ for different values of $\beta$ (with $\gamma = 1$). From Figure \ref{fig:corr}, the correlation function can be approximated as an exponential function of the form
$$
C(r) \simeq e^{-r/t_0},
$$
where the constant $t_0$ gives the correlation time. Figures \ref{fig:corr}b-c show that the correlation time is largely independent of $\beta$, and that it is approximately given by $1/\gamma$.

\begin{figure*}[htbp]
\centering
\includegraphics[width=16cm]{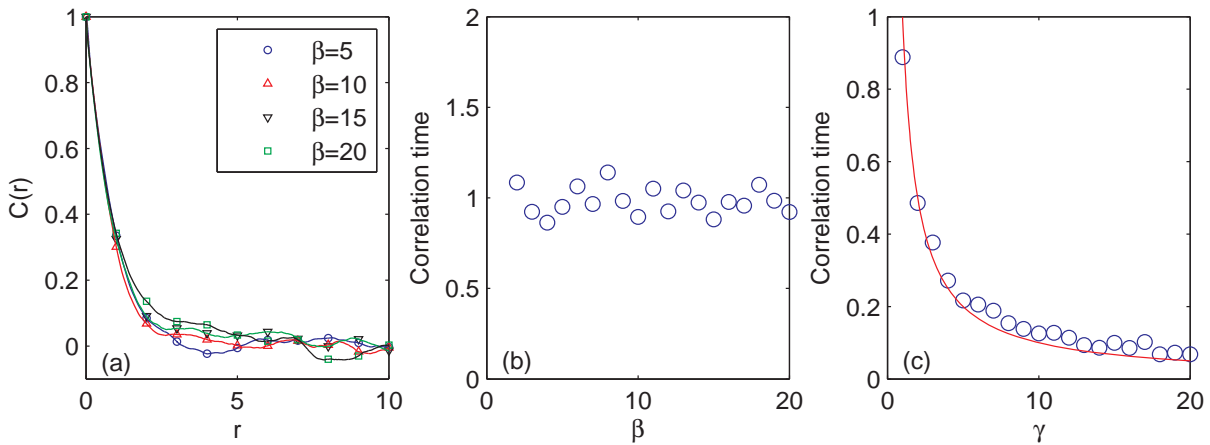}
\caption{(Color online) (a) Correlation function $C(r)$.  Here, $\gamma =   1$, and $\beta = 5$ (blue circles), $10$ (red up triangles), $15$ (black down triangles), $20$ (green squares), respectively. (b) Correlation time as a function of $\beta$ (with $\gamma = 1$).  (c) Correlation time as a function of $\gamma$ (with $\beta = 20$), solid curve is the fit  with $t_0 =
1/\gamma$.}
\label{fig:corr}
\end{figure*}

\subsection{Quasi-Gaussian distribution}\label{ss:quasigaussian}

From our numerical results in Sections \ref{sec:3.2} and \ref{sec:3.3}, it is clear that the excess kurtosis $\gamma_2$ of one of the irregular solutions of \eqref{eq:1} varies with $\beta$ and $\gamma$ according to $\gamma_2 \simeq -\gamma/\beta$. Thus, the
distribution approaches a Gaussian-like distribution when $\beta$ is large (and $\gamma$ is fixed), but one with a truncated tail so that it is supported on a set of finite measure. We call such a truncated Gaussian distribution a \textit{quasi-Gaussian distribution}, and consider these further in this section.

Let $\mu$ and $\sigma$ be the mean and standard deviation of a quasi-Gaussian noise, and assume that the noise signal is supported on an interval $[\mu - K, \mu + K]$.  Then the density function is given by
\begin{equation}
\label{eq:den1}
p(\vel; \mu, \sigma, K) = \left\{
\begin{array}{ll}
\mathcal{C} e^{-\frac{(\vel-\mu)^2}{2\sigma^2}}
& \qquad \mathrm{if} \quad |\vel-\mu|\leq K\\
0 & \qquad   \mathrm{ other\ wise,}
\end{array}\right.
\end{equation}
where
\begin{equation}
\mathcal{C}= \dfrac{1}{\sqrt{2\pi}\sigma (\Phi(K/\sigma) - \Phi(-K/\sigma))}
\end{equation}
and
\begin{equation}
\Phi(z) = \int_{-\infty}^z e^{-s^2/2} d s = \dfrac{1}{2}\left[1 +\mathrm{erf}\left(\dfrac{z}{\sqrt{2}}\right)\right].
\end{equation}
In particular, when $\mu=0$ and $\sigma = 1$, we have a standard quasi-Gaussian distribution, with density function
\begin{equation}
\label{eq:den}
p(\vel; 0, 1, K_0) = \left\{
\begin{array}{ll}
\mathcal{C}_0 e^{-\vel^2/2},& |\vel| \leq K_0\\
0, & \mathrm{other\ wise}.
\end{array}\right.
\end{equation}
where
\begin{equation}
\mathcal{C}_0 = \dfrac{1}{\sqrt{2\pi}\int_{-K_0}^{K_0} e^{-s^2/2} d s}.
\end{equation}
There is only one adjustable parameter, namely the bound $K_0$, in a standard quasi-Gaussian distribution.

We can now compare the distribution function obtained from our simulation data with the quasi-Gaussian distribution. To this end, we first normalized the signal sequence $\{v_n\}$. From Section \ref{sec:3.3}, let $\zeta_n = \vel_n/\sigma(\beta,\gamma)$ so the sequence $\{\zeta_n\}$ has mean $\mu = 0$, standard deviation $\sigma = 1$, and is bounded by
\begin{equation}
\label{eq:K}
K_0 = \dfrac{K(\beta,\gamma)}{\sigma(\beta,\gamma)} \simeq  \dfrac{\sqrt{\beta/\gamma}}{0.21 \sqrt{\beta/\gamma} + 0.19}.
\end{equation}
Equation \eqref{eq:K} gives the relation between the equation parameter $\beta/\gamma$ and the adjustable distribution parameter $K_0$, which is shown in Figure \ref{fig:fit}a. We have $K_0\simeq 5$ when $\beta/\gamma$ is large.

Figure \ref{fig:fit}b shows the result of fitting the density function Eq. \eqref{eq:den} with our simulation data.  We can see that \eqref{eq:den} provides a
reasonable fit for the simulation data when $\beta \geq 6$ (and $\gamma = 1$).  Thus, the simulation results indicate that when $\beta/\gamma$ is large, the density of the distribution of irregular solution trajectories of the differential delay equation \eqref{eq:1} can be approximated by a quasi-Gaussian distribution.

\begin{figure}[htbp]
\centering
\includegraphics[width=8cm]{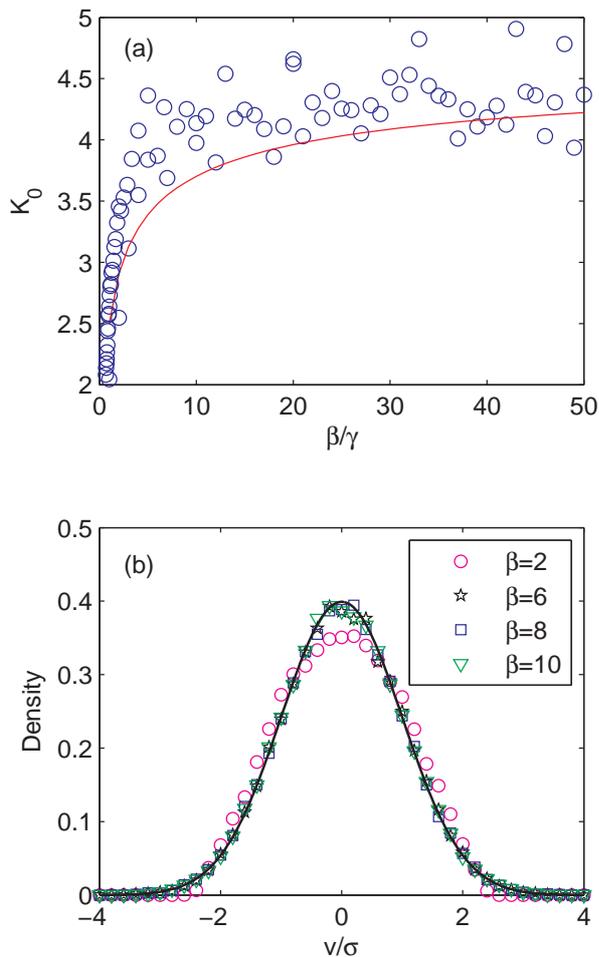}
\caption{(Color online) (a) The normalized bound $K$ as a function of $\beta/\gamma$. The solid curve shows the fit with $\sqrt{\beta/\gamma}/(0.21 \sqrt{\beta/\gamma} + 0.19)$. Circles are data from simulations in Sections \ref{sec:3.2} and \ref{sec:3.3}. (b) Density functions obtain from a solution of the differential delay equation with given value of $\beta$ (and with $\gamma = 1$). Solid curve shows the density function of quasi-Gaussian distribution according to Eq. \eqref{eq:den} and with $K_0$ obtained from $\beta = 10$.}
\label{fig:fit}
\end{figure}

\section{Deterministic Brownian Motion}
\label{sec:4}

Now, we will show that the differential delay equation \eqref{eq:bm}, in a suitable parameter region, can generate dynamics with many of the properties of Brownian motion in spite of the fact that the evolution equation is deterministic. Therefore, these dynamics are examples of \textit{deterministic Brownian motion}.

Consider solutions of the following deterministic system:
\begin{equation}
\label{eq:w0}
\left\{
\begin{array}{rcl}
\dfrac{d x}{dt}&=& \vel\\
\dfrac{d \vel}{dt} &=&  - \gamma \vel + \sin(2\pi\beta \vel(t-1))
\end{array}
\right.
\end{equation}
Here $x(t)$ measures the position of a particle with velocity $\vel(t)$. Figure \ref{fig:process}a shows sample solutions $x(t)$ of \eqref{eq:w0},  which are akin to the dynamics of a Brownian particle.  Figure \ref{fig:process}b shows the density functions of $\Delta x(t)$, the displacement from the particle initial position at different times   $t$. These numerical results show that at any time $t$, $\Delta x(t)$ has Gaussian like distribution.

\begin{figure}[htbp]
\centering
\includegraphics[width=8cm]{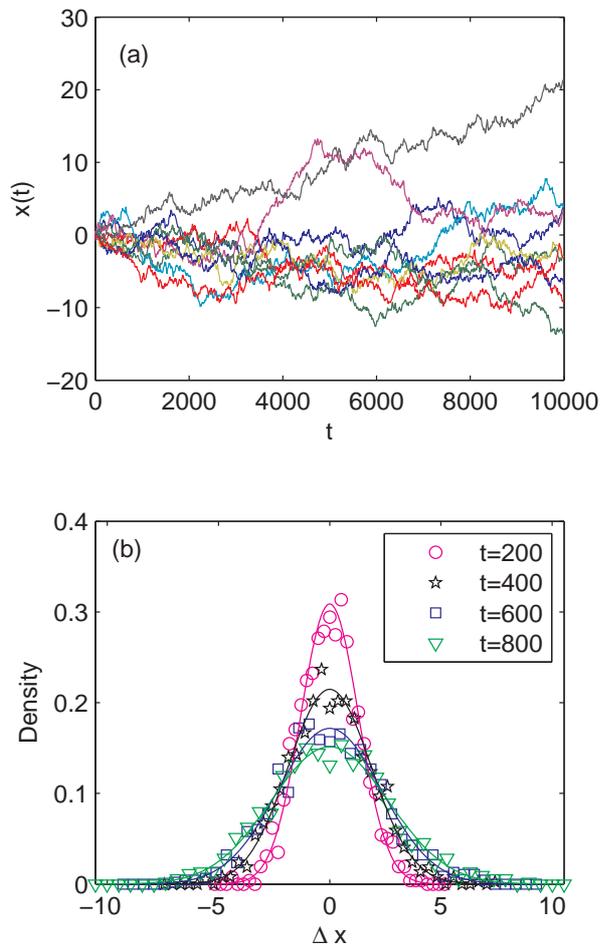}
\caption{(Color online) (a) The numerically produced deterministic Brownian motion $x(t)$.  Here $\beta=20, \gamma = 1$. Initial velocities are constants as in Eq. \eqref{eq:ini}. (b) Density functions of $\Delta x(t)$ ($= x(t)  - x(0)$) at different times $t$.  The symbols are taken from the numerical solutions, while the solid curves show the density function of the corresponding quasi-Gaussian distribution. }
\label{fig:process}
\end{figure}

Figure \ref{fig:MSD} shows $\langle [\Delta x(t)]^2\rangle$, the dependence of mean square displacement (MSD), as a function of $t$. In the simulations, we set $\beta = 20$, and chose different values of $\gamma$ (as shown in the figure panel). For each $\gamma$, the  MSD $\langle [\Delta x(t)]^2\rangle$ is obtained from $10^3$ independent trajectories, each with a randomly selected constant initial velocity. In Figure \ref{fig:MSD}, we normalized the results for different parameters through $D = \sigma^2/\gamma$, the ``diffusion constant''.  Our simulations show that $\langle [\Delta x(t)]^2 \rangle = 2 D t$ at long time scales, as predicted by Einstein't theory for Brownian motion \cite{Einstein:05}. At short time scales, we have $\langle [\Delta x(t)]^2\rangle = c D t^2$, where the pre-factor $c$ depends on the initial condition. This result agrees well with recent measurements of Brownian motion using an optical tweezer \cite{Raizen}.  In the interpretation of experimental data, it is a long debated question if and by what mens we can distinguish whether an observed irregular signal is deterministically chaotic or stochastic \cite{Cencini2000,Cencini2010}. Results in the current study indicate that experimentally observed data for Brownian motion can be reproduced by solutions of a deterministic differential delay equation over a wide range of time scales (six orders of magnitude).  Thus, differential delay equations provide an alternative way for reproducing  ``random'' signals.

\begin{figure}[htbp]
\centering
\includegraphics[width=8cm]{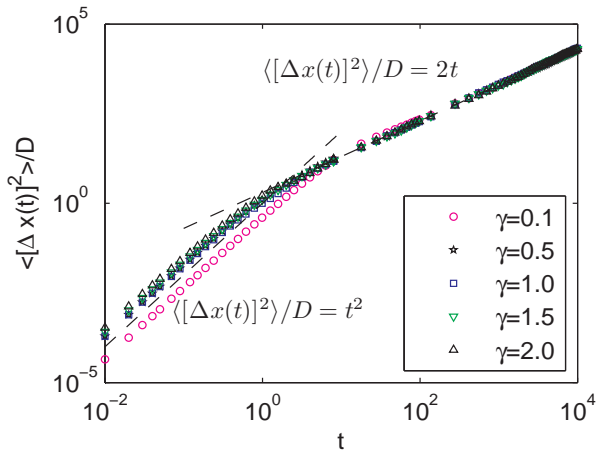}
\caption{(Color online) Mean square displacement (MSD) $\langle [\Delta x(t)]^2\rangle$ of deterministic Brownian motions. Here $D = \sigma^2/\gamma$, which is analogous to the diffusion constant of a Brownian particle in solution. The dashed lines show $\langle [\Delta x(t)]^2\rangle/D = t^2$ at short times,  and $\langle [\Delta x(t)]^2 \rangle = 2 t$ for a longer time scale respectively. Here $\beta=20$ and $\gamma$ are shown in the figure panel.}
\label{fig:MSD}
\end{figure}

\section{Conjectures}\label{conjectures}

The results that we have presented to this point are so intriguing that we are led to formulate a series of conjectures.  Though we believe these to be true, all efforts to prove them have proved fruitless to date.  We present them in the hope that others will find their proof a challenge that they are able to overcome.

In formulating these conjectures, we focus on the irregular solutions for large $\beta$, and therefore we  will always assume that
$\beta$ is such that equation \eqref{eq:1} has no stable steady state or stable periodic solution. In particular, according to Theorem
\ref{th:bif}, we will always assume $\gamma=1$ and $\beta\in I$, with $I$ defined by \eqref{eq:defI}.

Let $\vel_\beta(t;\phi)$   be the solution of
\begin{equation}
\left\{
\begin{array}{l}
\dfrac{d \vel}{dt} = - \vel + \sin(2\pi \beta \vel(t-1)),\\
\vel(t) = \phi(t),\quad -1\leq t\leq 0.
\end{array}\right.
\end{equation}
Define
\begin{eqnarray}
\mu(\beta; \phi) &=& \lim_{T\to\infty} \dfrac{1}{T}\int_0^T\vel_\beta(t;\phi) d t,\\
K(\beta;\phi) &=& \lim_{T\to \infty}\sup_{0\leq t\leq T} |\vel_\beta(t;\phi)|,\\
\sigma(\beta; \phi)&=& \lim_{T\to \infty}\sqrt{\dfrac{1}{T}\int_0^T \vel_\beta(t;\phi)^2 d t},\\
\mu_4(\beta;\phi)&=& \lim_{T\to \infty}\dfrac{1}{T}\int_0^T \vel_\beta(t;\phi)^4 d t.
\end{eqnarray}

In these conjectures, we always assume that the initial function $\phi(t)$ is taken such that the solution $\vel_\beta(t;\phi)$ is not a steady
state solution, i.e., $\phi$ satisfies the condition:
\begin{equation}
\label{eq:cond}
-\phi(t) + \sin(2\pi\beta \phi(t))\not\equiv 0,\quad -1\leq t\leq 0.
\end{equation}

We then have the following conjectures.

\begin{conjecture}
\label{conj:1}
For any $\phi\in C([-1,0], \mathbb{R})$ that satisfies \eqref{eq:cond}, we have
\begin{equation}
\label{eq:con1}
\mu (\beta; \phi) = 0
\end{equation}
for any $\beta\in I$.
\end{conjecture}

\textit{Remark 1} When $\beta \not\in I$, we have to exclude the cases in which $\vel_\beta(t;\phi)$ converges to either a stable steady state or a stable periodic solution arising through a Hopf bifurcation. Thus, with this exclusion, the solution $\vel_\beta(t;\phi)$ will converge to either a periodic solution which is symmetric about $0$, or an irregular solution. Equation \eqref{eq:con1} always holds for a symmetric solution. Therefore, to prove Conjecture \ref{conj:1} we only need to show that \eqref{eq:con1} is satisfied by the irregular solutions for any $\beta>0$.

\begin{conjecture}
\label{conj:2}For any $\phi\in C([-1,0],\mathbb{R})$ satisfying \eqref{eq:cond}, the limit
\begin{equation}
\label{eq:con2}
\lim_{\beta\in I, \beta\to\infty} \beta^{1/2}K(\beta;\phi)
\end{equation}
exists, independent of $\phi$, and is positive.
\end{conjecture}

\begin{conjecture}
\label{conj:3}
For any $\phi\in C([-1,0],\mathbb{R})$ satisfying \eqref{eq:cond}, the limit
\begin{equation}
\label{eq:con3}
\lim_{\beta\in I, \beta\to\infty} \beta^{1/2}\sigma(\beta;\phi)
\end{equation}
exists, independent of $\phi$, and is positive.
\end{conjecture}

\begin{conjecture}
\label{conj:4}For any $\phi\in C([-1,0],\mathbb{R})$ satisfying  \eqref{eq:cond}, we have
\begin{equation}
\label{eq:con3-1}
\lim_{\beta\in I, \beta\to\infty} \dfrac{\mu_4(\beta;\phi)}{\sigma^4(\beta;\phi)} = 3.
\end{equation}
\end{conjecture}

From Conjectures \ref{conj:2} and \ref{conj:3}, the constant
\begin{equation}
\label{eq:K0-1}
K_0 = \lim_{\beta\in I, \beta \to\infty} \dfrac{K(\beta; \phi)}{\sigma(\beta;\phi)}
\end{equation}
is well defined for any $\phi\in C([-1,0],\mathbb{R})$ satisfying  \eqref{eq:cond}, and independent of $\phi$.  Conjecture \ref{conj:4} suggests that when $\beta\in I$ is sufficiently large, the density of the distribution of the time series $\vel_\beta(t;\phi)$ tends to a Gaussian with mean $\mu = 0$, and standard deviation
$\sigma(\beta; \phi)$, but is truncated at $\pm K(\beta;\phi)$.  Therefore, let
\begin{equation}
P_\beta(z;\phi) = \lim_{T\to\infty}\dfrac{1}{T}\int_0^T H\left(z \sigma(\beta;\phi)- \vel_\beta(t;\phi)\right) d t
\end{equation}
where $H(\cdot)$ is the Heaviside step function. Then $P_\beta(z;\phi)$ measures the probability that $\vel_\beta(t;\phi) < z \sigma (\beta;\phi)$.

\begin{conjecture}
\label{conj:5}
Let $K_0$ be given by \eqref{eq:K0-1} and $p(\vel;0,1,K_0)$ be defined by \eqref{eq:den}. For any $\phi\in C([-1,0],\mathbb{R})$ satisfying \eqref{eq:cond}, we have
\begin{equation}
\lim_{\beta\in I, \beta\to\infty}P_\beta(z;\phi)  = \int_{-\infty}^z p(\vel;0,1,K_0) d \vel
\end{equation}
for all $z\in \mathbb{R}$.
\end{conjecture}

Conjecture \ref{conj:5} can be thought of as a Central Limit Theorem result for the irregular solutions of the differential delay equation
\eqref{eq:1}.

These conjectures were based on our numerical studies of the differential delay equation
\begin{equation}
\label{eq:non}
\dfrac{d v}{dt} = - v + F(v(t-1))
\end{equation}
with the nonlinear function $F(v)$ taken to be a sinusoidal function. We suspect that the same results   also hold for any bounded and oscillating nonlinear function such that the solution is ``chaotic''. For example, Figure \ref{fig:fH} shows the numerical results for the step function nonlinearity
\begin{equation}
\label{eq:step}
F(v) = 2 \left[H(\sin(2\pi \beta v)) - \dfrac{1}{2}\right]
\end{equation}
and the quasi-periodic nonlinearity
\begin{equation}
\label{eq:quasi}
F(v) = \dfrac{1}{2}\left[\sin (2\pi \beta v) + \sin (2\sqrt{2}\pi \beta v)\right],
\end{equation}
respectively. Therefore, the proposed conjectures  may be universal for these deterministic ``chaotic'' dynamics.

\begin{figure}[htbp]
\centering
\includegraphics[width=8cm]{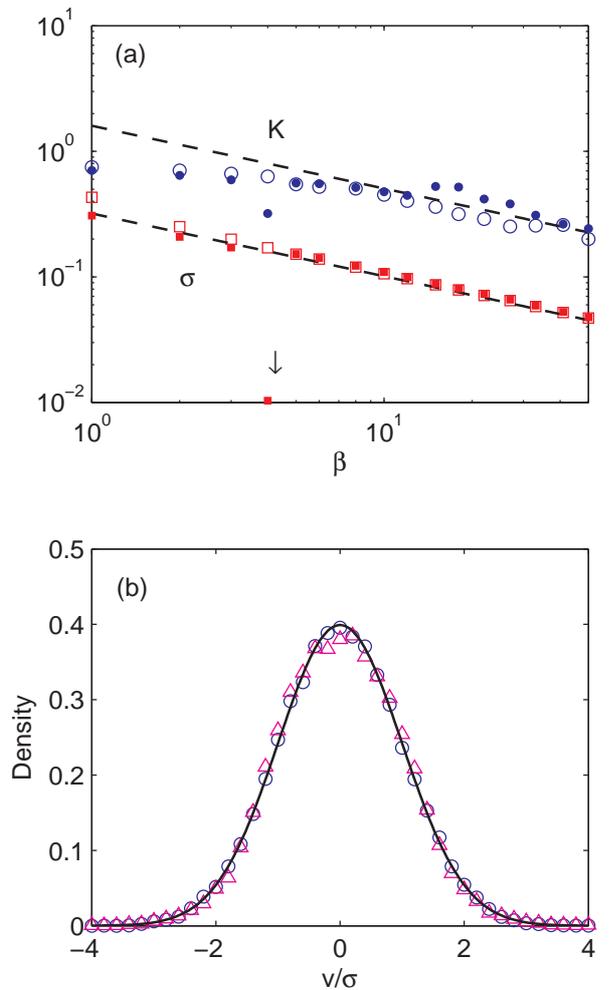}
\caption{(Color online) (a) Upper bound $K$ and standard deviation $\sigma$ for solutions of \eqref{eq:non} with the step function \eqref{eq:step} (hollow) and quasi-periodic function (solid) nonlinearities respectively. The dashed lines show the dependence $K \sim
\beta^{-1/2}$ and $\sigma \sim \beta^{-1/2}$. In the case of the quasi-periodic function, the result for $\beta=4$ (marked by an arrow) is exceptional because the solutions are not chaotic (in fact, they are periodic solutions). (b) Normalized distributions obtained from the numerical solutions of \eqref{eq:non} with a step function (blue circles) and quasi-periodic function (magenta triangles) nonlinearity respectively. The solid curve is the density function of the quasi-Gaussian distribution. Here $\beta = 10$.}
\label{fig:fH}
\end{figure}

\section{Discussion and conclusions}
\label{sec:5}

In this paper, we have studied a simple differential delay equation that displays a variety of behaviors, including chaotic solutions.

In Section \ref{sec:2} we  carried out a complete bifurcation analysis for the
steady state solutions. When $\gamma = 1$, our analysis show that for any positive integer $k$, there is an interval $I_k = (a_k, b_k)$ that contains $k+1/4$,
such that when $\beta \in I_k$, there are two stable steady states (which are symmetric with respect to $0$). Furthermore, $(\beta, \vel) = (a_k, \pm y_k)$ are
saddle node bifurcation points, and $(\beta, \vel) = (b_k, \pm z_k)$ are Hopf bifurcation points. Explicit expressions for $a_k, b_k, y_k, z_k$ are given in Section
\ref{sec:2.1}.  When $\beta$ increases past $b_k$, two stable periodic solutions are generated at the Hopf bifurcation points. In additional to
these regular solutions, when $\beta > 0.85$, the equation also has irregular solutions, which show chaotic behaviors.

In Section \ref{sec:3}, we numerically studied the probabilistic properties of the irregular (chaotic) solutions.  Our simulations suggest that when $\beta$ is
large ($\gamma = 1$), the discrete sequences $\{\vel_n\}$ generated by irregular solutions $\vel(t)$  (obtained by sampling each numerical solution every 1000 steps, i.e., $\vel_n =\vel(n\times 1000 \Delta t)$) have the character of Gaussian distributed noise, but are truncated at the bound
$ \pm K$  which varies as $\beta^{-1/2}$. The variance of the time series $\{\vel_n\}$ also depends on $\beta$ as
$\beta^{-1/2}$.  When $\beta$ is sufficiently large,  the density of the distribution of the normalized solution approaches a quasi-Gaussian distribution
\eqref{eq:den} with parameter $K_0 = K/\sigma \simeq 5$.

In Section \ref{sec:3}, the quasi-Gaussian distribution was obtained from the time series $\{\vel_n\}$ of solutions $v(t)$ of Eq. \eqref{eq:1} with constant initial functions, and each solution is sampled every 1000 steps. We also noted that the stationary density function is independent of the initial function. We argue that the main results obtain in Section \ref{sec:3} are independent of the sampling frequency. Thus, Figure \ref{fig:mathp} shows the bound $K$ and standard deviation $\sigma$, for different values of $\beta$, of the time series $\{\vel_n\}$ when we sample the numerical solutions of Eq. \eqref{eq:1} every 1 step ($\vel_n = \vel(n \Delta t)$). The results obviously show $K\sim \beta^{-1/2}$ and $\sigma\sim \beta^{-1/2}$, as we have seen in Section \ref{sec:3}. Figure \ref{fig:mathp}b shows the distribution obtained from the time series obtained by sampling a solution with different frequencies (every 1 step, 500 steps, and 1000 steps, respectively). These simulation results suggested several conjectures (which we have been unable to prove) for the probabilistic properties of the solutions of \eqref{eq:1} as given in Section \ref{conjectures}.

\begin{figure}[htbp]
\centering
\includegraphics[width=8cm]{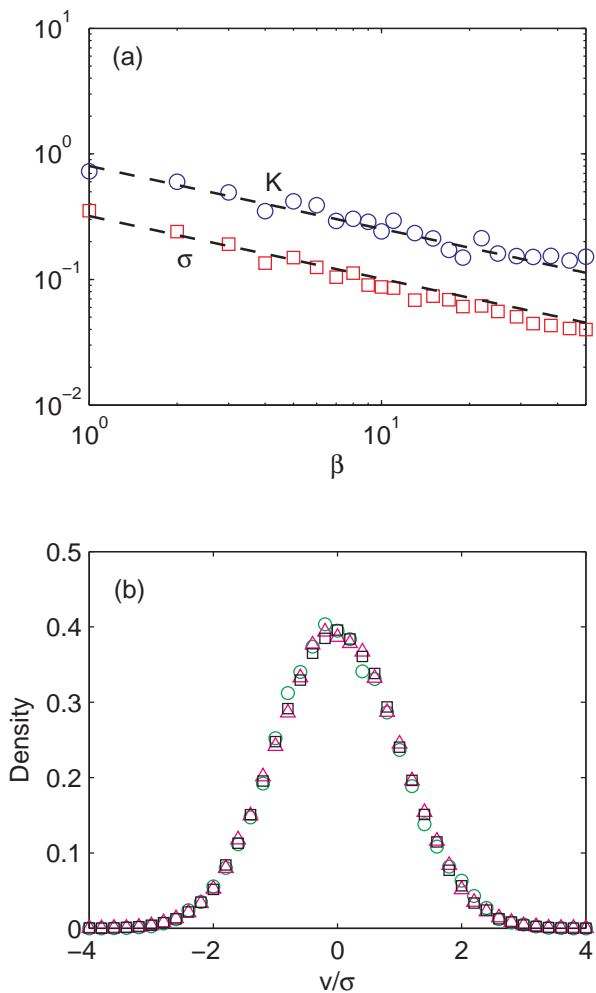}
\caption{(Color online) (a) Upper bound $K$ and standard deviation $\sigma$ obtained from time series $\{\vel_n\}$ obtained by sampling numerical solutions every 1 step, i.e., $\vel_n = \vel(n \Delta t)$. Dashed lines show the dependence $K \sim
\beta^{-1/2}$ and $\sigma \sim \beta^{-1/2}$. Here $\gamma = 1$. (b) Normalized distributions obtained from time series obtained by sampling a solution with different frequencies (every 1 step (green circles), every 500 steps (black squares), and every 1000 steps (magenta triangles)). Here $\beta = 10, \gamma = 1$.}
\label{fig:mathp}
\end{figure}

Section \ref{sec:4} has shown that a Brownian motion like behavior can be reproduced from the quasi-Gaussian distributed solution of the differential delay equation.  This deterministic Brownian motion shows behavior similar to that of experimentally observed Brownian motion, and therefore provides an alternative way to model apparently erratic behavior in nature. For example, the close to 50\% efficiency exhibited in certain biological processes is very difficult to explain from a purely thermodynamic point of view \cite{Bockries1993}. The dynamical alternative presented here could afford another possible explanation, which could originate in a  coherent, dynamical behavior at the molecular level of description. We feel that the application of the concept of a {\it deterministic Brownian motion} in modeling physical or biological phenomena that display stochastic aspects will be of great interest in future studies.

Finally, in Section \ref{conjectures} we have formulated five conjectures derived from our extensive numerical studies of this paper.  We hope that these serve as a challenge to others.

The significance of these results is, we feel, interesting.  All experimental measurements typically exhibit fluctuations around some value, and it is customary (indeed the norm) to interpret these as ``noise" and the implicit assumption is that these fluctuations are due to some random process that has no deterministic origin.  The density of the distribution of these fluctuations is, moreover, typically approximately Gaussian distributed but they are never truly Gaussian distributed (in the sense that the density is supported on the entire real line) but are always quasi-Gaussian in the sense that we have used it here.  The numerical studies that we have presented lend strong circumstantial support to the alternative interpretation that what is typically held to be the signature of a random ({\it i.e.} non-deterministic) process could equally well be the signature of a completely deterministic process \cite{Cencini2000,Cencini2010}.  The same implications were pointed out by Mackey and Tyran-Kami\'{n}ska \cite{Mackey06} based on analytic computations in a similar situation.

\appendix*

\section{Proof of Theorem \ref{th:bif}}

To prove theorem \ref{th:bif}, we first need the following two lemmas, which are obvious from the implicit function theorem, and the proofs are omitted.

\begin{lemma}
\label{le:1}
Let $(\beta_0, \vel_0)$ satisfy \eqref{eq:ss}, and assume that
$$1 - 2\pi \beta_0 \cos (2\pi \beta_0 \vel_0) \not = 0.
$$
Then there is a function $g(\beta)$, such that:
\begin{enumerate}
\item[\textnormal{(1)}] The function $g(\beta)$ satisfies $g(\beta_0) = \vel_0$, and
$$g(\beta) = \sin (2\pi \beta g(\beta))$$
for $\beta$ in a neighborhood of $\beta_0$.
\item[\textnormal{(2)}] The function $g(\beta)$ is differentiable, and
\begin{equation}
\dfrac{d g(\beta)}{d \beta} = \dfrac{2\pi g(\beta)\cos(2\pi \beta g(\beta))}{1 - 2 \pi \beta \cos(2\pi\beta g(\beta)}.
\end{equation}
\item[\textnormal{(3)}] The function $g(\beta)$ satisfies
\begin{equation}
\label{eq:dg}
\dfrac{d(\beta \cos (2\pi \beta g(\beta)))}{d\beta}  = \dfrac{\cos (2\pi \beta g(\beta)) - 2\pi \beta}{1 - 2\pi \beta \cos (2\pi \beta
g(\beta))}
\end{equation}
for $\beta$ in a neighborhood of $\beta_0$.
\end{enumerate}
\end{lemma}

\begin{lemma}
\label{le:2}
Let $(\beta_0, \vel_0)$ satisfy \eqref{eq:ss}, and assume that
$$
2\pi \vel_0 \cos (2\pi \beta_0 \vel_0) \not = 0.
$$
Then there is a function $h(\vel)$, such that:
\begin{enumerate}
\item[\textnormal{(1)}] The function $h(\vel)$ satisfies $h(\vel_0) =\beta_0$, and
$$h(\vel) = \sin (2\pi h(\vel) \vel)$$
for $\vel$ in a neighborhood of $\vel_0$.
\item[\textnormal{(2)}] The function $h(\vel)$ is differentiable, and
\begin{equation}
\dfrac{d h(\vel)}{d \vel} = \dfrac{1 - 2 \pi h(\vel) \cos(2\pi h(\vel) \vel)}{2\pi \vel\cos(2\pi h(\vel) \vel)}.
\end{equation}
\item[\textnormal{(3)}] The function $h(\vel)$ satisfies
\begin{equation}
\label{eq:df}
\dfrac{d( h(\vel) \cos (2\pi h(\vel) \vel))}{d \vel}  = \dfrac{\cos (2\pi h(\vel) \vel) - 2\pi h(\vel)}{ 2 \pi \vel \cos(2\pi h(\vel) \vel)}
\end{equation}
for $\vel$ in a neighborhood of $\vel_0$.
\end{enumerate}
\end{lemma}

\textit{Proof of Theorem \ref{th:bif}.}\  Points (1), (2) and (3) are  obvious.

Point (4). Here, we only consider the positive solutions $\vel$, as the negative solutions are symmetric.  The saddle node bifurcation points are given by solutions of the equations
\begin{equation}
\label{eq:akx}
\vel = \sin (2\pi \beta \vel),\quad 2\pi \beta\cos(2\pi \beta \vel) = 1.
\end{equation}
When $\vel > 0$ is a solution of \eqref{eq:akx}, then $\sin(2\pi \beta \vel) > 0$ and $\cos(2\pi\beta \vel) > 0$.  Therefore solutions of \eqref{eq:akx} will always satisfy
$$2k \pi < 2\pi \beta \vel < 2 k \pi + \dfrac{\pi}{2},\quad k \in \mathbb{N}_0.$$
Hence, equations \eqref{eq:akx} are equivalent to
\begin{equation}
\vel^2 + \left (\dfrac{1}{2\pi \beta}\right )^2 = 1,\quad 2\pi \beta \vel = 2 k \pi + \arccos \dfrac{1}{2\pi \beta},\quad k\in \mathbb{N}_0.
\end{equation}
Consequently, the bifurcation points are given by the solutions of
\begin{equation}
\label{eq:ak}
\left\{
\begin{array}{rcl}
2\pi k &=& 2\pi \beta \sqrt{1 - \left (\dfrac{1}{2\pi \beta}\right )^2} - \arccos \dfrac{1}{2\pi \beta},\\
\vel& =& \dfrac{k}{\beta} + \dfrac{1}{2\pi \beta} \arccos \dfrac{1}{2\pi \beta}
\end{array} \right. \quad k\in\mathbb{N}_0.
\end{equation}
From Lemma \ref{le:0}, these equations have a unique solution, which gives the saddle node bifurcation points $(\beta, \vel) = (a_k, y_k)$.

Point (5). The Hopf bifurcation points are given by the solutions of
\begin{equation}
\label{eq:x}
\vel = \sin (2\pi \beta \vel), \quad 2 \pi \beta \cos (2 \pi \beta \vel) = \sec\omega.
\end{equation}
Similar to our previous argument in Point (4), we only consider positive solutions $\vel$ that are given by the solutions of
\begin{equation}
\label{eq:bk}
\left\{
\begin{array}{rcl}
2\pi k &=& 2\pi \beta \sqrt{1 - \left (\dfrac{\sec \omega}{2\pi \beta}\right )^2} - \arccos \dfrac{\sec \omega}{2\pi \beta},\\
\vel & =& \dfrac{k}{\beta} + \dfrac{1}{2\pi \beta} \arccos \dfrac{\sec\omega}{2\pi \beta}
\end{array} \right. \quad k\in\mathbb{N}_0,
\end{equation}
which give the Hopf bifurcation points $(b_k, z_k)$.

Point (6). Let
$$F(\beta, \vel) = \vel - \sin(2\pi \beta \vel).$$
For any $k\in \mathbb{N}^*$, our previous arguments indicate that $(a_k, y_k)$ satisfies
$$F(a_k, y_k) = 0,$$
and further
$$\dfrac{\partial F(a_k, y_k)}{\partial \beta} = - 2\pi y_k \cos(2\pi a_k y_k) \not= 0.$$
Thus, from Lemma \ref{le:2}, there is a function $\beta = h_k(\vel)$, such that $a_k = h_k(y_k)$, and it is differentiable in a neighborhood of $y_k$.

We will show that the function $h_k(\vel)$ can be continued to the interval $\vel\in (0,y_k]$.  If not, there is $\vel^* \in (0, y_k]$ and $\beta^*$ such
that
$$F(\beta^*, \vel^*) = 0,$$
and
$$\dfrac{\partial F(\beta^*, \vel^*)}{\partial \beta} = - 2\pi \vel^* \cos(2\pi \beta^* \vel^*) = 0,$$
which implies $\cos (2\pi \beta^* \vel^*) = 0$. Therefore we should have
$$(\vel^*)^2 = (\sin (2\pi \beta^* \vel^*))^2 = 1.$$
However, this is impossible since $y_k < 1$. Thus, we conclude that the function $h_k(\vel)$ can be continued to the entire interval $(0, y_k]$, and further that
$\cos(2\pi h_k(\vel) \vel) > 0$ for any $\vel\in (0, y_k]$.

Next, we will show that $h_k(\vel)$  is a decreasing function for $\vel\in (0, y_k]$. From \eqref{eq:df}, we have
$$
\dfrac{d (h_k(\vel) \cos(2\pi h_k(\vel) \vel))}{d v} = \dfrac{\cos (2\pi h_k(\vel) \vel) - 2 \pi h_k(\vel)}{2\pi \vel \cos (2\pi h_k(\vel) \vel)} < 0.
$$
Thus, for $\vel\in (0, y_k)$, we have
$$
2 \pi h_k(\vel) \cos(2 \pi h_k(\vel) \vel) > 2 \pi a_k \cos (2\pi a_k y_k) = 1.
$$
Note that $h_k(\vel) > a_0 = 1/(2\pi)$ and $\cos (2\pi h_k(\vel) \vel) > 0$. Therefore,  from Lemma \ref{le:2}
$$
\dfrac{d h_k(\vel)}{d \vel} = \dfrac{1  - 2 \pi h_k(\vel) \cos (2\pi h_k(\vel) \vel)}{2 \pi \vel \cos(2\pi h_k(\vel) \vel)} < 0.
$$

Now, the function $\beta = h_k(\vel)$ is well defined and decreasing for $\vel\in (0, y_k]$. Thus, the inverse function, denoted by $\vel = f_k(\beta)$, is also
well defined, continuous at $\beta \in [a_k, \infty)$, and such that $\vel=\pm f_k(\beta)$ satisfy \eqref{eq:ss} and \eqref{eq:u1}.  From \eqref{eq:u1},
it is easy to conclude that the steady state solutions $\vel(t)\equiv \pm f_k(\beta)$ are unstable, and Point (6) is proved.

Point (7). For any $k\in \mathbb{N}_0$, the Hopf bifurcation point $(b_k, z_k)$ satisfies
$$
1 - 2\pi b_k \cos (2\pi b_k z_k) = 1 - \sec \omega > 0.
$$
Therefore, we can apply Lemma \ref{le:1},  and there is a function $g_k(\beta)$ such that $\vel = g_k(\beta)$ satisfies \eqref{eq:ss}, and $z_k =
g_k(b_k)$.

When $\beta > b_k\ (> a_0)$, from \eqref{eq:dg} in Lemma \ref{le:1}, we have
$$
\begin{array}{l}
\dfrac{d (1 - 2\pi \beta \cos (2\pi \beta g_k(\beta)))^2}{d \beta}\\
 \quad = - 4  \pi (\cos (2\pi \beta g_k(\beta)) - 2 \pi \beta) > 0.
 \end{array}
$$
Thus,  we have $\partial F(\beta,v)/\partial v\not=0$ for $\beta > b_k$ and the function $g_k(\beta)$ can be continued to $\beta\in (b_k, \infty)$, and the steady state solutions $\vel(t) \equiv \pm g_k(\beta)$ are
unstable.

When $\beta < b_k$, we will show that the function $g_k(\beta)$ can be continued to $\beta \in (a_k, b_k)$. If not, there is $\beta^*\in (a_k, b_k)$
such that
$$
F(\beta^*, \vel^*) = 0
$$
and
$$
\dfrac{\partial F(\beta^*, \vel^*)}{\partial \vel} = 1 - 2\pi \beta^* \cos (2\pi \beta^* \vel^*) = 0.
$$
Therefore, we must have $\beta^* = a_{k'}$ for some $k'\in \mathbb{N}_0$. However, $a_k$ is the maximum of such values that are less than $b_k$, and thus we must have $\beta^* = a_k$.  Therefore, the function $g_k(\beta)$ can be continued to $\beta \in (a_k, b_k)$, and $g_k(a_k) = y_k$. These arguments show that the function $g_k(\beta)$ is well defined in the interval $(a_k, b_k)$, and satisfies $y_k = g_k(a_k), z_k = g_k(b_k)$.

Now, we only need to show that when $a_k < \beta < b_k$, the steady state solutions $\vel(t)\equiv \pm g_k(\beta)$  are stable. Since $(1 - 2\pi \beta
\cos (2\pi \beta g_k(\beta)))^2$ is increasing with respect to $\beta$, and
$$
\left\{
\begin{array}{l}
1 - 2 \pi a_k \cos (2\pi a_k g_k(a_k)) = 0, \\
1 - 2 \pi b_k \cos (2\pi b_k g_k(b_k)) = 1 - \sec\omega > 0,
\end{array}\right.
$$
we have
$$
 0 < 1 - 2\pi \beta \cos(2\pi \beta g_k(\beta)) < 1 - \sec \omega,
$$
{\it i.e.},
$$
\sec\omega < 2\pi \beta \cos (2\pi \beta g_k(\beta)) < 1
$$
for any $\beta \in (a_k, b_k)$.  Therefore the steady state solutions
$\vel(t)\equiv g_k(\beta)$ are stable, and Point (7) is proved.

Point (8) is obvious from the above arguments, and the theorem is proved.

\section*{Acknowledgements}  We are grateful for research support from NSERC (Canada) (JL and MCM), MITACS (Canada), and the Alexander von Humboldt Stiftung (Germany) (MCM).  We are especially indebted to Dr. Catherine Foley (Montreal) for initial numerical experiments on the system
\eqref{eq:non}-\eqref{eq:step}, and to Dr. hab. Marta Tyran-Kami\'nska (Katowice) for extensive earlier discussions about this problem.  This research was carried out in Montreal, Bremen (Germany) and Beijing and MCM would like to thank the Zhou Pei-Yuan Center for Applied Mathematics of Tsinghua University for their hospitality.

\bibliographystyle{aipnum4-1}

%

\end{document}